\documentclass[11pt,envcountsect]{llncs}
\usepackage{amsmath,amssymb,times,a4wide}

\def\citeN{\cite}

\def\tfloor#1{{\lfloor#1\rfloor}}

\def\Prob{\mathrm{I\!Pr}}
\def\Exp{\mathrm{I\!E}}

\newcommand{\reals}{\mathrm{I\!R}}

\def\cost{\mathrm{cost}}

\def\SC{\mathrm{SC}}
\def\MC{\mathrm{MC}}

\def\EC{\mathrm{EC}}
\def\PL{\mathrm{PtL}}

\def\C{\mathcal{C}}

\def\IntLength{\ell}
\def\IntStart{\alpha}
\def\Slope{\lambda}

\title{Strategyproof Facility Location for Concave Cost Functions%
\thanks{This research was supported by the project AlgoNow, co-financed by the European Union (European Social Fund - ESF) and Greek national funds, through the Operational Program ``Education and Lifelong Learning'' of the National Strategic Reference Framework (NSRF) - Research Funding Program: THALES, investing in knowledge society through the European Social Fund.}}

\author{Dimitris Fotakis\inst{1} \and Christos Tzamos\inst{2}}

\institute{%
School of Electrical and Computer Engineering,\\
National Technical University of Athens, 157 80 Athens, Greece.\\
Email: {\tt fotakis@cs.ntua.gr}
\and
Computer Science and Artificial Intelligence Laboratory,\\
Massachusetts Institute of Technology, Cambridge, MA 02139.\\
Email: {\tt tzamos@mit.edu}}

\begin{document}
\maketitle

\begin{abstract}
We consider $k$-Facility Location games, where $n$ strategic agents report their locations on the real line, and a mechanism maps them to $k$ facilities. Each agent seeks to minimize his connection cost, given by a nonnegative increasing function of his distance to the nearest facility. Departing from previous work, that mostly considers the identity cost function, we are interested in mechanisms without payments that are (group) strategyproof for any given cost function, and achieve a good approximation ratio for the social cost and/or the maximum cost of the agents.

We present a randomized mechanism, called {\sc Equal Cost}, which is group strategyproof and achieves a bounded approximation ratio for all $k$ and $n$, for any given concave cost function. The approximation ratio is at most $2$ for {\sc Max Cost} and at most $n$ for {\sc Social Cost}. To the best of our knowledge, this is the first mechanism with a bounded approximation ratio for instances with $k \geq 3$ facilities and any number of agents. Our result implies an interesting separation between deterministic mechanisms, whose approximation ratio for {\sc Max Cost} jumps from $2$ to unbounded when $k$ increases from $2$ to $3$, and randomized mechanisms, whose approximation ratio remains at most $2$ for all $k$. On the negative side, we exclude the possibility of a mechanism with the properties of {\sc Equal Cost} for strictly convex cost functions. We also present a randomized mechanism, called {\sc Pick the Loser}, which applies to instances with $k$ facilities and only $n = k+1$ agents. For any given concave cost function, {\sc Pick the Loser} is strongly group strategyproof and achieves an approximation ratio of $2$ for {\sc Social Cost}.
\end{abstract}


{\bf Keywords:} Algorithmic Mechanism Design; Social Choice; Facility Location Games

\thispagestyle{empty}%
\setcounter{page}{0}%
\newpage
\pagestyle{plain}
\pagenumbering{arabic}

\section{Introduction}
\label{s:intro}

We consider \emph{$k$-Facility Location games}, where $k$ facilities are placed
on the real line based on the preferences of $n$ strategic agents. Such problems are motivated by natural scenarios in Social Choice, where the government plans to build a fixed number of public facilities in an area (see e.g., \cite{Miy01}). The choice of the locations is based on the preferences of local people, or \emph{agents}. Each agent reports his ideal location, and the government applies a (deterministic or randomized) \emph{mechanism} that maps the agents' preferences to $k$ facility locations.

The agents evaluate the outcome of the mechanism according to their \emph{connection cost}, given by a nonnegative increasing function $c(d)$ of the distance $d$ of their ideal location to the nearest facility.
Agents seek to minimize their connection cost, and may misreport their ideal locations in an attempt of manipulating the mechanism. Therefore, the mechanism should be \emph{strategyproof}, i.e., should ensure that no agent can benefit from misreporting his location, or even \emph{group strategyproof}, i.e., should ensure that for any coalition of agents misreporting their locations, at least one of them does not benefit.
The government's goal is to minimize an objective function of the agents' connection cost. Most prominent among them are the objective of {\sc Social Cost}, which considers the total cost of the agents, and the objective of {\sc Max Cost}, which considers the maximum cost of an agent.
So, in addition to (group) strategyproofness, the mechanism should either optimize or achieve a reasonable approximation to the designated objective function, thus ensuring that the outcome is socially efficient.

\smallskip\noindent{\bf Previous Work.}
The numerous applications and the elegance of the model have attracted a significant volume of research on the problem. In Social Choice, the emphasis has been on characterizing the class of (group) strategyproof mechanisms for locating a single facility if the agents' preferences are \emph{single-peaked}. Roughly speaking, an agent has single-peaked preferences if he has an ideal location (or \emph{peak}), and consistently prefers less the locations farther from it. However, the strength of his preference for locations closer to his peak is not explicitly quantified by any function of the distance.
For general single-peaked preferences, a classical result of Moulin \cite{Moul80} shows that the class of deterministic strategyproof mechanisms for locating a single facility on the line coincides with the class of generalized median mechanisms (see also the surveys of Barber\'{a} \cite{Bar01} and Sprumont \cite{Spru95}, and \cite[Chapter~10]{AGT-book}). Schummer and Vohra \cite{SV02} extended this characterization to tree metrics, and proved that for non-tree metrics, any onto strategyproof mechanism must be a dictatorship. More recently,  Dokow et al. \cite{DFMN12} obtained similar characterizations for locating a single facility on the discrete line and on the discrete circle.

Adopting an optimization viewpoint to Facility Location games, Procaccia and Tennenholtz \cite{PT09} introduced the framework of \emph{approximate mechanism design without money}. The basic idea is to consider game-theoretic versions of optimization problems, such as $k$-Facility Location, where efficiency is quantified by an objective function (instead of efficiency related properties, such as onto, non-dictatorship, and Pareto-efficiency, typically studied in Social Choice). Then, any reasonable approximation to the optimal solution can be regarded as a socially desirable outcome, and one seeks to determine the best approximation ratio achievable by strategyproof mechanisms.
As for the preferences of the agents, with respect to which strategyproofness is defined, this line of research adopted the standard definition of Facility Location problems from Operations Research (see e.g., \cite{MF90}). Thus, it implicitly abandoned the setting of general single-peaked preferences, in favor of the more restricted (and technically easier to handle) case where the agents' cost is given by a linear function $c(d) = \alpha d$ of their distance $d$ to the nearest facility.
Translated into this framework, the results of \cite{Moul80,SV02} imply a deterministic strategyproof mechanism that minimizes the {\sc Social Cost} for 1-Facility Location on the line and in tree metrics. On the negative side, the impossibility result of \cite{SV02} implies that the best approximation ratio achievable for the objective of {\sc Social Cost} by deterministic strategyproof mechanisms for 1-Facility Location in general metrics is $n-1$.
However, the explicit quantification of agents' preferences now allows for randomized mechanisms that are strategyproof with respect to the agents' expected cost (a.k.a. incentive compatible in expectation, see e.g., \cite[Section~9.5.6]{AGT-book}) and
may achieve better approximation ratios.

Since \cite{PT09}, there has been a considerable interest in quantifying the best approximation ratio achievable by strategyproof mechanisms for $k$-Facility Location on the line and in general metric spaces. As a result, the approximability of $k$-Facility Location (with linear cost functions) by deterministic and randomized strategyproof mechanisms has become well understood in many interesting cases (see also Fig.~\ref{fig:summary}). The main message is that deterministic strategyproof mechanisms can only achieve a bounded approximation ratio if we have at most $2$ facilities \cite{PT09,FT12}. On the other hand, randomized mechanisms achieve better approximation ratios for $2$-Facility Location, and also a bounded approximation ratio if we have $k \geq 2$ facilities and only $n = k+1$ agents \cite{EGTPS11}. Notably, such instances are known to be hard for deterministic mechanisms. In particular, the inapproximability of $k$-Facility Location by anonymous deterministic strategyproof mechanisms, for all $k \geq 3$, was proved in \cite{FT12} for instances with only $n = k+1$ agents.

\begin{figure}[t]%
\centerline{\begin{tabular}{c}
{\sc Max Cost}\\ 
\hspace*{-1mm}\begin{tabular}{|l|c|c|c|c|} \hline
	&\hspace*{5mm}$k = 1$\hspace*{5mm}%
    &\hspace*{15mm}$k = 2$\hspace*{15mm}%
    &\hspace*{5mm}$2 < k < n-1$\hspace*{5mm}%
    &\hspace*{8mm}$k = n-1$\hspace*{8mm}\\ \hline
Deterministic\ \ \ & $2$\ \cite{PT09} & $2$\ \cite{PT09}
              & $\infty$\ \cite{FT12}
              & $\infty$\ \cite{FT12}\\ \hline
Randomized\ & $1.5$\ \cite{PT09} & $[1.5, 5/3]$\ \cite{PT09}
              & $[1.5, \mathbf{2}]$\ [{\bf here}]
              & $1.5$\ \cite{EGTPS11}\\ \hline
\end{tabular}\medskip\\
{\sc Social Cost}\\
\hspace*{-0.8mm}\begin{tabular}{|l|c|c|c|c|} \hline
	&\hspace*{5mm}$k = 1$\hspace*{5mm}%
    &\hspace*{15mm}$k = 2$\hspace*{15mm}%
    &\hspace*{5mm}$2 < k < n-1$\hspace*{5mm}%
    &\hspace*{8mm}$k = n-1$\hspace*{8mm}\\ \hline
Deterministic\ \ \ & $1$\ \cite{Moul80} & $n-2$\ \cite{FT12}, \cite{PT09}
              & $\infty$\ \cite{FT12}
              & $\infty$\ \cite{FT12}\\ \hline
Randomized\ &  $1$\ \cite{Moul80}
              & $[1.045, 4]$\ \cite{LWZ09}, \cite{LSWZ10}
              & $[1.045, \mathbf{n}]$\ [{\bf here}]
              & $[1.045, \mathbf{2}]$\ [{\bf here}]\\ \hline
\end{tabular}
\end{tabular}}
\caption{\label{fig:summary}Summary of known results on the approximability of $k$-Facility Location on the line (with linear cost functions) for the objectives of {\sc Max Cost} and {\sc Social Cost}. In each cell, we have either the precise approximation ratio (if known) or the interval determined by the best known lower and upper bounds. In cells with two references, the first is for the lower bound and the second for the upper bound. We note that the lower bound on the approximation ratio of deterministic mechanisms for $k \geq 3$ is only shown for anonymous mechanisms. The randomized upper bounds proven in this work are shown in bold and hold for any concave cost function.}
\end{figure}

\smallskip\noindent{\bf Motivation and Contribution.}
Our work is motivated by two natural questions related to approximate mechanism design without money for $k$-Facility Location. The first question is about the approximability of $k$-Facility Location by randomized strategyproof mechanisms for instances with any number of facilities and any number of agents. Prior to this work, we have only known randomized mechanisms with a bounded%
\footnote{The approximation ratio of a mechanism for $k$-Facility Location is bounded if it is a function of $n$ and $k$. We highlight that this property is essentially objective-independent, since any mechanism with a bounded approximation ratio for e.g., {\sc Max Cost} also has a bounded approximation for {\sc Social Cost} and for the objective of minimizing the $L_p$ norm of the agents' costs, for any $p \geq 1$, and vice versa.}
approximation ratio if we have either at most 3 facilities or $k$ facilities and only $n = k+1$ agents. Most importantly, all the randomized upper bounds in Fig.~\ref{fig:summary} are achieved by mechanisms that balance between strategyproofness and efficiency using different approaches (see e.g., \cite{PT09,LSWZ10,EGTPS11}).

The second question is whether the restriction to linear cost functions is a necessary price to pay for adopting the elegant optimization framework of Procaccia and Tennenholtz \cite{PT09} and aiming at a reasonable approximation ratio. In fact, we can imagine a few natural scenarios where the agents' cost is best described by a convex or a concave non-decreasing cost function $c(d)$ of their distance $d$ to the nearest facility. For example, a convex cost function captures the fact that the growth rate of the people's disutility from commuting increases with the distance (e.g., in addition to cost and time considerations, people get more and more tired if they commute over long distances). On the other hand, a concave cost function captures the fact that the growth rate of the traveling time decreases with the distance (e.g., people walk over short distances, bike over medium distances, drive over long distances, and take a plane over really long ones). To a certain extent, a setting where the agents' cost function is not fixed, but is given as part of the input, would be closer to the setting of general single-peaked preferences in Social Choice. Then, a mechanism should be strategyproof, or even group strategyproof, for any given cost function $c$, just as generalized median mechanisms are strategyproof for any collection of single-peaked preferences, while the approximation ratio may also depend on some quantitative properties (e.g., the derivative) of $c$. Notably, this holds for the class of percentile mechanisms \cite{SBS12}, which decide on the facility locations based on the ordering of the agents on the line, are group strategyproof, and include the optimal (wrt. the approximation ratio for linear cost functions) deterministic mechanisms for $1$ and $2$-Facility Location on the line. However, percentile mechanisms have an unbounded approximation ratio for all $k \geq 3$. In contrast, the strategyproofness of known randomized mechanisms crucially depends on the linearity of the cost function (see e.g., \cite[Mechanism~1]{PT09} which is not strategyproof e.g., for $c(d) = \sqrt{d}$).

In this work, we make significant progress in both research directions above. Our main technical contribution consists of two randomized mechanisms, called {\sc Equal Cost} and {\sc Pick the Loser}, that are group strategyproof and achieve a bounded approximation ratio for any number of facilities and any given concave cost function.

{\sc Equal Cost}, presented in Section~\ref{s:equal-cost}, applies to instances with any number of facilities $k$ and any number of agents $n$, and is the first (group) strategyproof mechanism with a bounded approximation ratio for all $k$ and $n$. Its approximation ratio is at most $2$ for {\sc Max Cost} and at most $n$ for {\sc Social Cost}, for all concave cost functions $c$. Combined with the lower bound of \cite{PT09} for the objective of {\sc Max Cost}, this implies that the best approximation ratio achievable by randomized mechanisms for $k$-Facility Location on the line and is at least $1.5$ and at most $2$, for all $k$ and for all concave cost functions. Moreover, we obtain an interesting separation between deterministic mechanisms, whose approximation ratio for {\sc Max Cost} jumps from $2$ to unbounded when $k$ increases from $2$ to $3$, and randomized mechanisms, whose approximation ratio remains a small constant for all $k$.

From a technical viewpoint, {\sc Equal Cost} works by equalizing the expected cost of all agents. The mechanism first covers the agents' locations with $k$ disjoint intervals of length $\IntLength$, where $\IntLength$ is chosen so that $c(\IntLength)$ is at most twice the optimal maximum cost of an agent. Then, taking the cost function $c$ into account, it computes a random variable $X$ in $[0, \IntLength]$, so that all locations $x \in [0, \IntLength]$ have the same expected cost, under $c$, if $x$ is connected to a facility distributed in $[0, \IntLength]$ according to $X$. Finally, {\sc Equal Cost} places a facility in each interval according to the random variable $X$ so that all agents have an expected cost equal to the expectation of $c(X)$.

The key technical claim in the analysis of {\sc Equal Cost} is that if the cost function $c$ is concave and piecewise linear, a random variable $X$ with the desired properties exists and can be computed efficiently as the solution to a homogeneous system of linear equations (Lemma~\ref{l:random-var}). This claim can be generalized to any continuous concave function, but the technical details have to do with techniques for the solution of integral equations and are beyond the scope of this work. We show that {\sc Equal Cost} is (resp. strongly) group strategyproof for any given (resp. strictly) concave cost function $c$, and that the agents' expected cost is at most the maximum cost of an agent in the optimal solution for the objective of {\sc Max Cost} (Lemma~\ref{l:ec-cost}). In addition to implying the approximation guarantees, the upper bound on the expected cost of the agents indicates that the facility allocation of {\sc Equal Cost} is fair in expectation, and does not unnecessarily increase the agents' disutility.

To demonstrate the natural behavior of {\sc Equal Cost} for typical cost functions, we derive the exact form of the random variable $X$ for three important cases: linear cost functions, piecewise linear cost functions with two pieces, and exponential cost functions of the form $c(d) = 1 - e^{-\lambda d}$ (Section~\ref{s:applications}). Moreover, we show how to implement {\sc Equal Cost} if the agents and the facilities should lie in a bounded interval (Section~\ref{s:limitations}). This implies that {\sc Equal Cost} can be applied to instances where the agents lie on a circle metric, with the same approximation guarantees, but rather surprisingly, with group strategyproofness carrying over only if the number of facilities is even.

On the negative side, we exclude the possibility of a mechanism with the properties of {\sc Equal Cost} for strictly convex cost functions (Section~\ref{s:convex}). Specifically, we show that the expected cost of the agents in the same interval cannot be equalized if the cost function $c$ is strictly convex. Moreover, employing an exponential cost function, we show (Lemma~\ref{l:convex}) that there does not exist a randomized strategyproof mechanism with a bounded approximation ratio for any given convex cost function (note that the approximation ratio here may also depend on the cost function).

In Section~\ref{s:loser}, we focus on the simpler and elegant setting where we have $k$ facilities and only $n = k+1$ agents. This setting was motivated and studied in \cite{EGTPS11}, and deserves special attention not only because such instances are among the hardest ones for deterministic mechanisms (see e.g., \cite[Theorem~7.1]{FT12}), but also because they make {\sc Equal Cost} perform poorly for the objective of {\sc Social Cost}. We present the {\sc Pick the Loser} mechanism that allocates facilities to all but a single agent, designated as the loser. The probability distribution according to which the loser is chosen is motivated by the probability distribution used by \citeN{Kouts11} for scheduling on selfish unrelated machines. Our key technical contribution here is to show that {\sc Pick the Loser} is strongly group strategyproof for any given concave cost function (Lemma~\ref{l:pal-strategyproofness}). We also show that {\sc Pick the Loser} achieves an approximation ratio of $2$ for the objective of {\sc Social Cost}. Thus, we significantly improve on the previously best known approximation ratio of $n/2$ achieved by the {\sc Inversely Proportional} mechanism of \citeN{EGTPS11} for this class of instances. Moreover, the small approximation ratio of {\sc Pick the Loser} nicely complements the poor performance of {\sc Equal Cost} for such instances.

\smallskip\noindent{\bf Other Related Work.}
For the objective of {\sc Max Cost}, Alon et al. \cite{AFPT09} almost completely characterized the approximation ratios achievable by randomized and deterministic mechanisms for 1-Facility Location in general metrics and rings.
For the objective of {\sc Social Cost}, Nissim et al. \cite{NST10} and Fotakis and Tzamos \cite{FT10} considered imposing randomized mechanisms that achieve an additive approximation of $o(n)$ and an approximation ratio of $4k$ for $k$-Facility Location on the line and in general metric spaces, respectively.
For $1$-Facility Location on the line and the objective of minimizing the $L_2$ norm of the agents' distances to the facility, Feldman and Wilf \cite{FW11} proved that the best approximation ratio is $1.5$ for randomized and $2$ for deterministic mechanisms. Moreover, they presented a class of randomized mechanisms that includes all known strategyproof mechanisms for 1-Facility Location on the line.

\section{Notation, Definitions, and Preliminaries}
\label{s:prelim}

For a random variable $X$, we let $\Exp[X]$ denote the \emph{expectation} of $X$. For an event $E$ in a sample space, we let $\Prob[E]$ denote the probability that $E$ occurs.

\smallskip\noindent{\bf Instances.}
We consider $k$-Facility Location with $k \geq 1$ facilities and $n \geq k+1$ agents on the real line.
We let $N = \{ 1, \ldots, n\}$ be the set of agents. Each agent $i \in N$ resides at a location $x_i \in \reals$, which is $i$'s \emph{private} information. An \emph{instance} is a tuple $(\vec{x}, c)$, where $\vec{x} = (x_1, \ldots, x_n) \in \reals^n$ is the agents' locations profile and $c: \reals_{\geq 0} \mapsto \reals_{\geq 0}$ is a cost function that gives the connection cost of each agent. The cost function $c$ is \emph{public knowledge} and the same for all agents. Normalizing $c$, we assume that $c(0) = 0$. If the cost function $c$ is clear from the context, we let an instance simply consist of $\vec{x}$.

For an $n$-tuple $\vec{x} = (x_1, \ldots, x_n) \in \reals^n$, we let
$\vec{x}_{-i} = (x_1, \ldots, x_{i-1}, x_{i+1}, \ldots, x_n)$
be $\vec{x}$ without $x_i$. For a non-empty set $S$ of indices, we let $\vec{x}_S = (x_i)_{i \in S}$ and $\vec{x}_{-S} = (x_i)_{i \not\in S}$.
We write 
$(\vec{x}_{-i}, a)$ to denote the tuple $\vec{x}$ with $a$ in place of $x_i$, $(\vec{x}_{-\{i,j\}}, a, b)$ to denote the tuple $\vec{x}$ with $a$ in place of $x_i$ and $b$ in place of $x_j$, and so on.

\smallskip\noindent{\bf Mechanisms.}
A \emph{deterministic mechanism} $F$ for $k$-Facility Location maps an instance $(\vec{x}, c)$ to a $k$-tuple $(y_1, \ldots, y_k) \in \reals^k$, $y_1 \leq \cdots \leq y_k$, of facility locations. We let $F(\vec{x}, c)$ (or simply $F(\vec{x})$, whenever $c$ is clear from the context) denote the outcome of $F$ for instance $(\vec{x}, c)$, and let $F_j(\vec{x}, c)$ denote $y_j$, i.e., the $j$-th smallest coordinate in $F(\vec{x}, c)$. We write $y \in F(\vec{x}, c)$ to denote that $F(\vec{x}, c)$ has a facility at location $y$.
A \emph{randomized mechanism} $F$ maps an instance $(\vec{x}, c)$ to a probability distribution over $k$-tuples $(y_1, \ldots, y_k) \in \reals^k$.

\smallskip\noindent{\bf Connection Cost, Social Cost, Maximum Cost.}
Given a $k$-tuple $\vec{y} = (y_1, \ldots, y_k)$, $y_1 \leq \cdots \leq y_k$, of facility locations, the connection cost of agent $i$ with respect to $\vec{y}$, denoted $\cost(x_i, \vec{y})$, is
\( \cost(x_i, \vec{y}) = c(\min_{1 \leq j \leq k} |x_i - y_j|) \).
Given a deterministic mechanism $F$ and an instance $(\vec{x}, c)$, we let $\cost(x_i, F(\vec{x}, c))$ (or simply, $\cost(x_i, F(\vec{x}))$, if $c$ is clear from the context) denote the connection cost of agent $i$ with respect to the outcome of $F(\vec{x}, c)$.
If $F$ is a randomized mechanism, the expected connection cost of agent $i$ is
\[ \cost(x_i, F(\vec{x}, c)) = \Exp_{\vec{y} \sim F(\vec{x}, c)}[\cost(x_i, \vec{y})] \]

The {\sc Max Cost} of a deterministic mechanism $F$ for an instance $(\vec{x}, c)$ is
\[ \textstyle \MC[F(\vec{x}, c)] = \max_{i \in N} \cost(x_i, F(\vec{x}, c)) \]
The expected {\sc Max Cost} of a randomized mechanism $F$ for an instance $(\vec{x}, c)$ is
\[ \textstyle \MC[F(\vec{x}, c)] =
     \Exp_{\vec{y} \sim F(\vec{x}, c)}[\max_{i \in N} \cost(x_i, \vec{y}) ] \]
The optimal {\sc Max Cost}, denoted $\MC^\ast(\vec{x}, c)$, is
\(  \MC^\ast(\vec{x}, c) = \min_{\vec{y} \in \reals^k} \max_{i \in N} \cost(x_i, \vec{y})  \).

The (resp. expected) {\sc Social Cost} of a deterministic (resp. randomized) mechanism $F$ for an instance $(\vec{x}, c)$ is
\( \SC[F(\vec{x}, c)] = \sum_{i = 1}^n \cost(x_i, F(\vec{x}, c)) \).
The optimal {\sc Social Cost}, denoted $\SC^\ast(\vec{x}, c)$, is %
\(  \SC^\ast(\vec{x}, c) = \min_{\vec{y} \in \reals^k} \sum_{i = 1}^n \cost(x_i, \vec{y})  \).


\smallskip\noindent{\bf Approximation Ratio.}
A (randomized) mechanism $F$ for $k$-Facility Location achieves an \emph{approximation ratio} of $\rho \geq 1$ for a class of cost functions $\C$ and the objective of {\sc Max Cost} (resp. {\sc Social Cost}), if for all cost functions $c \in \C$ and all location profiles $\vec{x}$,
\( \MC[F(\vec{x}, c)] \leq \rho\,\MC^\ast(\vec{x}, c) \)
(resp. \( \SC[F(\vec{x}, c)] \leq \rho\,\SC^\ast(\vec{x}, c) \)\,).

\smallskip\noindent{\bf Strategyproofness and Group Strategyproofness.}
A mechanism $F$ is \emph{strategyproof} for a class of cost functions $\C$ if no agent can benefit from misreporting his location. Formally, $F$ is strategyproof if for all cost functions $c \in \C$, all location profiles $\vec{x}$, any agent $i$, and all locations $y$,
\[ \cost(x_i, F(\vec{x}, c)) \leq \cost(x_i, F((\vec{x}_{-i}, y), c))\,. \]

A mechanism $F$ is (weakly) \emph{group strategyproof} for a class of cost functions $\C$ if for any coalition of agents misreporting their locations, at least one of them does not benefit. Formally, $F$ is (weakly) \emph{group strategyproof} if for all cost functions $c \in \C$, all location profiles $\vec{x}$, any non-empty coalition $S$, and all location profiles $\vec{y}_S$ for $S$, there exists some agent $i \in S$ such that
\[ \cost(x_i, F(\vec{x}, c)) \leq \cost(x_i, F((\vec{x}_{-S}, \vec{y}_S), c))\,. \]

A mechanism $F$ is \emph{strongly group strategyproof} for a class of cost functions $\C$ if there is no coalition $S$ of agents misreporting their locations where at least one agent in $S$ benefits and the other agents in $S$ do not lose from the deviation. Formally, $F$ is strongly group strategyproof if for all cost functions $c \in \C$ and all location profiles $\vec{x}$, there do not exist a non-empty coalition $S$ and a location profile $\vec{y}_S$ for $S$, such that for all $i \in S$,
\[ \cost(x_i, F(\vec{x}, c)) \geq \cost(x_i, F((\vec{x}_{-S}, \vec{y}_S), c))\,, \]
and there exists some agent $j \in S$ with
\[ \cost(x_i, F(\vec{x}, c)) > \cost(x_i, F((\vec{x}_{-S},\vec{y}_S), c))\,. \]

\section{The {\sc Equal-Cost} Mechanism}
\label{s:equal-cost}

In this section, we present and analyze the {\sc Equal Cost} mechanism. At the conceptual level, {\sc Equal Cost}, or $\EC$, in short, works by equalizing the expected cost of all agents. Given an instance $(\vec{x}, c)$ of $k$-Facility Location on the line, $\EC$ works as follows:

\begin{description}
\item[Step 1] It computes an optimal covering of all agent locations with $k$ disjoint intervals $[\IntStart_i, \IntStart_i+\IntLength]$ that minimizes the interval length $\IntLength$ (wlog., we assume that $\IntStart_i < \IntStart_{i+1}$).

\item[Step 2] It constructs a random variable $X(\IntLength) \in [0,\IntLength]$ such that all locations $x \in [0,\IntLength]$ have the same the expected connection cost $\Exp[ c(|x-X|) ]$.

\item[Step 3] For every interval $[\IntStart_i,\IntStart_i+\IntLength]$, $\EC$ places a facility at $\IntStart_i+X$, if $i$ is odd, or at $\IntStart_i+\IntLength-X$, if $i$ is even.
\end{description}

We proceed to establish the main properties of $\EC$, summarized by the following theorem. For the proof, we examine, in the following sections, each step of the mechanism separately.

\begin{theorem}\label{th:equal-cost}
For the class of all concave cost functions, {\sc Equal Cost} is group strategyproof and achieves an approximation ratio of $2$ for the objective of {\sc Max Cost}, and an approximation ratio of $n$ for the objective of {\sc Social Cost}. Moreover, for every instance ($\vec{x}, c)$, with $c$ concave, and every agent $i$, $\cost(x_i, \EC(\vec{x}, c)) \leq \MC^\ast(\vec{x}, c)$.
\end{theorem}

\subsection{Step 1: Partitioning the Instance in Intervals}
\label{s:partitioning}

We can compute the minimum feasible interval length $\IntLength$ by checking all possible candidate values. The value of $\IntLength$ is equal to the distance $x_j-x_i$ for some agent locations $x_j \leq x_i$. So, there are at most $n^2/2$ candidate values for $\IntLength$. For each candidate value $\IntLength'$, we can check feasibility and compute a covering of all locations in $\vec{x}$ with intervals of length $\IntLength'$ as follows:

\begin{quote}
While there are uncovered agents, find the leftmost uncovered agent $i$, and create a new interval $[x_i, x_i + \IntLength']$.
\end{quote}

The above algorithm computes the minimum number of intervals of length $\IntLength'$ to cover $\vec{x}$. If this number is at most $k$, we set $\IntLength = \IntLength'$. We can also speed up the algorithm by binary search over the space of candidate values.

We observe that the partitioning into intervals of length $\IntLength$ is closely related to the optimal maximum cost $\MC^\ast(\vec{x}, c)$. In fact, an optimal solution can be obtained by placing a facility at the midpoint of each interval. Thus, the cost of the optimal solution is $\MC^\ast(\vec{x}, c) = c(\IntLength/2)$.

\subsection{Step 2: Constructing the Random Variable}

We next show that for any given cost function $c$, we can construct a family of random variables $X(\IntLength) \in [0,\IntLength]$ such the expected cost of every point in $[0,\IntLength]$ is the same. For convenience, we denote this cost as $C(\IntLength)$. We note that $C(\IntLength) = \Exp[ c(  | {X(\IntLength)-x} | ) ]$, for all $x \in [0,\IntLength]$. In particular, for $x = 0$, we get $C(\IntLength) = \Exp[c(X(\IntLength))]$.

We assume that the cost function $c$ is piecewise-linear with pieces of length $1$ and growth rates $\Slope_0, \Slope_1, \ldots, \Slope_i, \ldots$, where $\Slope_i$ is the growth rate in the interval $[i, i-1)$. For all $i$, $\Slope_i > 0$ and $\Slope_i \geq \Slope_{i+1}$, because $c$ is strictly increasing and concave. Our result applies to general concave functions either by discretizing appropriately, or by solving a continuous analog of the homogeneous linear system below through an integral equation. The technical details are related to the solution of integral equations and are beyond the scope of this work.

The support $S$ of the random variable $X(\IntLength)$ is every point $i$ and $\IntLength-i$, for integer $i = 0, \ldots, \tfloor{\IntLength}$. We note that if $\IntLength$ is an integer, we have only $|S| = \IntLength + 1$ points in the support, instead of $|S| = 2 (\tfloor{\IntLength} + 1)$ points in general.
The crucial observation is that the derivative of the expected cost function in every interval between consecutive points in the support must be $0$. So, to compute the probability $p_j$ assigned to each point $j$ in the support of $X(\IntLength)$, we write a set of $|S|-1$ linear equations and $|S|$ unknowns (the probability $p_j$ of each point $j$ in the support) requiring that the derivative of the expected cost function in each interval is $0$. So, we get the homogeneous linear system $\Lambda \vec{p} = \vec{0}$. If $\IntLength$ is an integer, the matrix $\Lambda$ is:
$$
\Lambda = \left( \begin{array}{ccccc}
\Slope_0 & -\Slope_0 & -\Slope_1 & \ldots & -\Slope_{\IntLength-1} \\
\Slope_1 & \Slope_0 & -\Slope_{0} & \ldots & -\Slope_{\IntLength-2} \\
\vdots & \vdots & \vdots & \vdots & \vdots \\
\Slope_{\IntLength-1} &  \Slope_{\IntLength-2} & \Slope_{\IntLength-3} & \ldots & -\Slope_{0}
 \end{array} \right) 
$$
Namely, the elements of the matrix $\Lambda$ are $\Lambda_{i, j} = \Slope_{i-j}$, if $i \geq j$, and $\Lambda_{i, j} = \Slope_{j-i-1}$, if $i < j$, for all $i = 0, \ldots, \IntLength-1$ and $j = 0, \ldots, \IntLength$, where $\Slope_\kappa$ denotes the growth rate of the piecewise-linear cost function $c$ at the support point $\kappa$.

If $\IntLength$ is not an integer, the elements of the matrix $\Lambda$ are $\Lambda_{i, j} = \Slope_{\tfloor{(i-j)/2}}$, if $i \geq j$, and $\Lambda_{i, j} = \Slope_{\tfloor{(j-i-1)/2}}$, if $i < j$, for all $i = 0, \ldots, 2\tfloor{\IntLength}$ and $j = 0, \ldots, 2\tfloor{\IntLength}+1$. Thus,
$$\Lambda = \left( \begin{array}{ccccccccccc}
\Slope_0 & -\Slope_0 & -\Slope_0 & -\Slope_1 & -\Slope_1 & -\Slope_2 &  -\Slope_2 & \ldots & -\Slope_{\lfloor \IntLength \rfloor-1}& -\Slope_{\lfloor \IntLength \rfloor-1}& -\Slope_{\lfloor \IntLength \rfloor} \\
\Slope_0 & \Slope_0 & -\Slope_0 & -\Slope_0 & -\Slope_1 & -\Slope_1 & -\Slope_2 &  -\Slope_2 & \ldots & -\Slope_{\lfloor \IntLength \rfloor-1}& -\Slope_{\lfloor \IntLength \rfloor-1} \\
\vdots & \vdots & \vdots & \vdots & \vdots & \vdots & \vdots & \vdots & \vdots & \vdots & \vdots \\
\Slope_{\lfloor \IntLength \rfloor}& \Slope_{\lfloor \IntLength \rfloor-1}& \Slope_{\lfloor \IntLength \rfloor-1} & \ldots & \ldots & \ldots & \Slope_1 & \Slope_1 & \Slope_0 & \Slope_0 & -\Slope_{0}
 \end{array} \right) 
$$

We now show that in both cases there is a unique symmetric probability distribution that satisfies the system of equations. For this purpose, we use the two lemmas below.
The first lemma is about a class of diagonally dominant matrices. It shows that we can bring any such matrix in a triangular form by performing Gaussian elimination, such that all diagonal elements are positive and all off-diagonal elements are less than or equal to $0$.

\begin{lemma}\label{l:gaussian}
Let $A$ be a $q \times n$, $q \leq n$ matrix so that $A_{i,i} > 0$, for all $i=1, \ldots, q$, $A_{i,j} \leq 0$, for all $i \neq j$, and $\sum_{i=1}^q A_{i,j} > 0$, for all $j=1, \ldots, q$. Then, by performing elementary row operations (Gaussian elimination) on $A$, we can get a row-echelon form $A'$ where $A'_{i,i} > 0$, for all $i=1, \ldots, q$, $A'_{i,j} = 0$, for all $i > j$, and $A'_{i,j} \leq 0$, for all $i < j$.
\end{lemma}

\begin{proof}
We use induction on $q$. The base case, where $q = 1$, is already in the desired form. Assuming that the lemma holds for $q \geq 1$, we show that it holds for $q+1$.

We have that
\( A = \left( \begin{array}{ccc}
a & \vec u^T \\
\vec v & B \end{array} \right) \),
with $a > 0$ and all elements of $\vec{u}$ and $\vec{v}$ non-positive.
With a single step of Gaussian elimination, we get
\( \left( \begin{array}{ccc}
a & \vec u^T \\
\vec 0 & B - \frac{\vec{v} \times \vec{u}^T}{a}  \end{array} \right) \).
To conclude the induction step, we show that the submatrix $B' = B - \frac {\vec{v} \times \vec{u}^T}{a}$ satisfies the properties of the lemma. Since all elements of $\vec{v} \times \vec{u}^T$ are non-negative, we still have $B'_{i,j} \leq 0$, for all $i \neq j$. So, we need to show that $\sum_{i=1}^{q} B'_{i,j} > 0$, for all columns $j=1, \ldots, q$, which also implies that $B'_{i,i} > 0$, for all $i = 1, \ldots, q$. For any column $j$, we have that:
\[ \sum_{i=1}^{q} B'_{i,j} = \sum_{i=1}^{q} (B_{i,j} - v_i  u_j / a)
 = \sum_{i=1}^{q} B_{i,j} - \tfrac{u_j}{a} \sum_{i=1}^{q} v_i
 > - u_j - \tfrac{u_j} a (-a) = 0\,. \]
For the last inequality, we use that $u_j \le 0$, and the hypothesis that $\sum_{i=1}^{q+1} A_{i,j} > 0$, which implies that $a + \sum_{i=1}^{q} v_i > 0$ and that $u_j + \sum_{i=1}^{q} B_{i,j} > 0$.
\qed\end{proof}

The next lemma shows that for the special class of matrices $\Lambda$ arising in our case, there is a solution to the homogeneous linear system $\Lambda \vec p = \vec 0$ that defines a probability distribution.

\begin{lemma}\label{l:random-var}
Let $A$ be a $n \times (n+1)$ matrix defined as $A_{i,j} = a_{i-j}$, where $a_{-n}, \ldots, a_{n-1}$ is a sequence of positive numbers such that $a_{m-1} = -a_{-m}$, for all $m \geq -n$, and $a_{m-1} \geq a_{m}$ for all $m \geq 1$. Then, the system $A \vec{p} = \vec{0}$ has a symmetric solution with $p_j = p_{n-j}$, $\sum p_j = 1$, and $p_j \geq 0$. Moreover, there is a unique symmetric solution $\vec{p}$ that satisfies these conditions.
\end{lemma}

\begin{proof}
We let $d_m = a_{m-1} - a_{m} \geq 0$, for $m \geq 1$. Then, the matrix $A$ can be written as:
\[
\left(\begin{array}{ccccc}
a_0 & -a_0 & -a_1 & \ldots & -a_{n-1} \\
a_1 & a_0 & -a_{0} & \ldots & -a_{n-2} \\
\vdots & \vdots & \vdots & \vdots & \vdots \\
a_{n-1} & \ldots & \ldots & \ldots & -a_{0}
 \end{array}\right)
\!=\!
\left(\!\begin{array}{ccccc}
a_0 & -a_0 & -a_0+d_1 & \ldots & {-a_0+\sum_{m=1}^{n-1} d_m} \\
a_0-d_1 & a_0 & -a_0 & \ldots & {-a_0+\sum_{m=1}^{n-2} d_m} \\
\vdots & \vdots & \vdots & \vdots & \vdots \\
{a_0-\sum_{m=1}^{n-1} d_m} & \ldots & \ldots & \ldots & -a_{0}
 \end{array}\right)\]

Taking the difference of every pair of $A$'s consecutive rows, we obtain the $(n-1)\times (n+1)$ matrix
\[A' =
\left( \begin{array}{ccccc}
-d_1 & 2 a_0 & -d_1 & \ldots & -d_{n-1} \\
-d_2 & -d_1 & 2 a_0  & \ldots & -d_{n-2} \\
\vdots & \vdots & \vdots & \vdots & \vdots \\
-d_{n-2} & -d_{n-3} & \ldots  & - d_1 & -d_2 \\
-d_{n-1} & -d_{n-2} & \ldots & 2 a_0 & -d_1
 \end{array} \right)
\]
To establish the lemma, we first use Lemma~\ref{l:gaussian} and show that (i) the nullspace of $A'$ contains a unique symmetric probability vector $\vec{p}$, and then show that (ii) the particular vector $\vec{p}$ is also in the nullspace of $A$.

As for claim (i), we first show that each coordinate $p_j$ of any vector $\vec{p}$ in the nullspace of $A'$ can be expressed as a non-negative linear combination of the coordinates $p_0$ and $p_n$. Formally, we show that for any coordinate $p_j$ of any solution $\vec{p}$ of $A' \vec{p} = \vec{0}$, there exist $\pi_j, \rho_j \geq 0$, such that $p_j = \pi_j p_0 + \rho_j p_n$.
To this end, we consider the $(n-1)\times (n+1)$ matrix $A''$, which is obtained from $A'$ by moving the first column of $A'$ to the end. We observe that $A''$ satisfies the conditions of Lemma~\ref{l:gaussian}, since $\sum_{m=1}^{n-1} d_m = a_0 - a_{n-1} < a_0$, and thus $2 a_0 - 2 \sum_{m=1}^{n-1} d_m > 0$. Hence, by applying Gaussian elimination to $A''$, we get a $(n-1)\times (n+1)$ matrix $G$ in a row-echelon form with $G_{i,i} > 0$, for all $i$, $G_{i,j} = 0$, for all $i > j$, and $G_{i,j} \leq 0$, for all $i < j$. Moreover, the nullspace of $A'$ essentially consists of the solutions $\vec{x}$ to the homogenous linear system $G \vec{x} = \vec{0}$. More precisely, any solution $\vec{x}$ of $G \vec{x} = \vec{0}$ corresponds to a solution $\vec{p}$ of $A' \vec{p} = \vec{0}$, where $p_0 = x_n$, $p_1 = x_0$, \ldots, $p_n = x_{n-1}$, and vice versa.

Due to the special form of $G$, we can find all solutions $\vec{x}$ of $G \vec{x} = \vec{0}$ by assigning values to the free variables $x_{n-1}$ and $x_n$ and performing backwards substitution so that we uniquely determine the values of the variables $x_0, \ldots, x_{n-2}$. Furthermore, due to the special form of $G$, this procedure results in expressing each variable $x_j$ as a non-negative linear combination of $x_{n-1}$ and $x_n$. Specifically, we can calculate $x_j$, for all $j = n-2, \ldots, 0$, from the equation $\sum_{i=0}^{n} G_{j,i} x_i = 0$. Solving for $x_j$, we get
\( x_j = -\sum_{i=j+1}^{n} G_{j,i} x_i / G_{j,j} \),
since $G_{j,j} > 0$ and $G_{j,i} = 0$, for all $j > i$. Moreover, all coefficients $-G_{j,i}/G_{j,j}$ are non-negative because $G_{j,i} \leq 0$, for all $j < i$, and $G_{j,j} > 0$. By induction, if every $x_{j'}$, $j' > j$, is a non-negative linear combination of $x_{n-1}$ and $x_{n}$, the same holds for $x_j$. Therefore, any coordinate $x_j$ of any solution $\vec{x}$ to $G \vec{x} = \vec{0}$ can be expressed as a non-negative linear combination of the free variables $x_{n-1}$ and $x_n$. Due to the aforementioned correspondence between the solutions $\vec{p}$ of $A' \vec{p} = \vec{0}$ and the solutions $\vec{x}$ of $G \vec{x} = \vec{0}$, we obtain that for any coordinate $p_j$ of any solution $\vec{p}$ to $A' \vec{p} = \vec{0}$, there exist $\pi_j, \rho_j \geq 0$, such that $p_j = \pi_j p_0 + \rho_j p_n$.

Hence, the nullspace of $A'$ is spanned by the vectors $\vec{p}^1$ and $\vec{p}^2$ determined by setting the free variables $p_0$ and $p_n$ to $(1,0)$ and to $(0,1)$, respectively. By the discussion above, all the coordinates of $\vec{p}^1$ and $\vec{p}^2$ are non-negative. To conclude the proof of claim (i), we observe that due to the symmetry of the homogeneous linear system $A'\vec{p} = \vec{0}$, we have that $\vec{p}^1_j = \vec{p}^2_{n-j}$, for all $j = 0, \ldots, n$. Therefore, there is unique symmetric vector in the nullspace of $A'$ with $L_1$ norm equal to $1$, namely the vector $\vec{p} = (\vec{p}^1 + \vec{p}^2) / |\vec{p}^1 + \vec{p}^2|_1$.

We proceed to show claim (ii), namely that the unique symmetric probability vector $\vec{p}$ in the nullspace of $A'$ is also in the nullspace of $A$. To this end, we define the $n\times n$ matrix
\[
M = \left( \begin{array}{ccccc}
1 & 0 & 0 & \ldots & 1 \\
-1 & 1 & 0 & \ldots & 0 \\
\vdots & \vdots & \vdots & \vdots & \vdots \\
0 & \ldots & -1  & 1 & 0 \\
0 & \ldots & 0 & -1 & 1
 \end{array} \right)
\]
We observe that the determinant of $M$ is equal to $2$, and thus $M$ is non-singular. Therefore, the linear system $A \vec{p} = \vec{0}$ is equivalent to the linear system $M A\,\vec{p} = \vec{0}$. So, we let $A_1$ and $A_n$ be the first and the last row of $A$, and further observe that $M A$ is a $n\times(n+1)$ matrix with its first row equal to $A_1+A_n$ and its remaining rows in one-to-one correspondence to the rows of $A'$. Since $\vec{p}$ is the unique symmetric probability vector satisfying $A'\vec{p} = \vec{0}$, we only need to show that $(A_1+A_n) \vec{p} = 0$, which follows immediately from the symmetry of $\vec{p}$. This completes the proof of claim (ii) and the proof of the lemma.
\qed\end{proof}

For every $\IntLength$, the homogeneous linear system $\Lambda \vec{p} = \vec{0}$ satisfies the conditions of Lemma \ref{l:random-var}. Hence, there exists a unique symmetric probability distribution $\vec{p}$ such that the expected cost $\Exp[c(|X(\IntLength) - x|)]$ is the same for every location $x \in [0,\IntLength]$. Next, we think of this unique symmetric solution $\vec{p}$ as a function of $\IntLength$, and establish a nice continuity property of it.

To this end, we fix an integer $m \geq 0$, and show that the random variable $X(\IntLength)$ converges in probability to the random variable $X(m)$, as $\IntLength \rightarrow m^+$. We observe that the linear system determining $\vec{p}$ is the same for all $\IntLength \in (m, m+1)$. So, we let $p^m_i$ be the probability assigned to each integer point $i$, $0 \leq i \leq m$. By symmetry, the probability assigned to each point $\IntLength-i$, $0 \leq i \leq m$, is also $p^m_i$. The limit $\lim_{\IntLength \rightarrow m^+} X(\IntLength) = \bar X$ is a random variable distributed according to a probability distribution that assigns probability $p^m_i+ p^m_{m - i}$ to each integer point $i$, $0 \leq i \leq m$. Since the distribution is symmetric and achieves the same expected cost for all points $x \in [0, m]$, it is, by Lemma~\ref{l:random-var}, the unique distribution with these properties. Therefore, we have that $X(m) = \bar X$. By the same argument, we can show that the random variable $X(\IntLength)$ converges in probability to the random variable $X(m+1)$, as $\IntLength \rightarrow (m+1)^-$.

By the continuity property above, the expected cost $C(\IntLength) = \Exp \left[ c \left( X(\IntLength) \right) \right]$ at each location $x \in [0,\IntLength]$ is a continuous function of $\IntLength$. Moreover, the discussion above implies that for all $\IntLength \in [m, m+1)$, $C(\IntLength) = \sum_{i=0}^{m} p_i^{m} (c(i) + c(\IntLength-i))$. Using these properties, we now show that $C(\IntLength)$ is an increasing function of $\IntLength$.

\begin{lemma}
The expected cost $C(\IntLength)$ is an increasing function of the interval length $\IntLength$.
\end{lemma}

\begin{proof}
Since $C$ is continuous, we only need to show that $C$ is increasing in each interval $[m, m+1)$, where $m \geq 0$ is any integer. To this end, we let $\IntLength \in [m, m+1)$, and consider any $\IntLength' \in (\IntLength, m+1)$. Then, we have that:
$$C(\IntLength) = E \left[ c ( X(\IntLength) ) \right] = \sum_{i=0}^{m} p_i^m (c(i) + c(\IntLength-i)) < \sum_{i=0}^{m} p_i^m (c(i) + c(\IntLength'-i)) = C(\IntLength')\,,$$
where the inequality holds because $\IntLength' > \IntLength$ and the cost function $c$ is increasing.
\qed\end{proof}

\subsection{Step 3: Establishing Group Strategyproofness}

We next prove that the random facility placement, in Step~3 of {\sc Equal Cost}, is group strategyproof.
The correlation of the facility placement, in Step~3, ensures that if an agent $j$ is located at $y$, his closest facility is always the one assigned to his closest interval. To justify this, let us consider any sample $x$ of the random variable $X$. We recall that the facilities are placed at $\IntStart_1 + x, \IntStart_2+\IntLength-x,\IntStart_3 + x, \ldots$. Let us assume that $\IntStart_i+\IntLength-x \leq y \leq \IntStart_{i+1}+x$. Then, the distance of $y$ to $\IntStart_i+\IntLength-x$ is $y-(\IntStart_i+\IntLength-x)$, while the distance of $y$ to $\IntStart_{i+1}+x$ is $\IntStart_{i+1}+x-y$. Hence, agent $j$ prefers the facility at interval $i$ if and only if $y-(\IntStart_i+\IntLength) < \IntStart_{i+1}-y$, i.e., the right endpoint of interval $i$ is closer to $y$ than the left endpoint of interval $i+1$.

To show that {\sc Equal Cost} is group strategyproof, we consider a coalition of agents $S$ that deviate to improve their cost. Let the original interval length, with respect to the true agents' locations, be $\IntLength$, and let the new interval length, after the deviation, be $\IntLength'$. We now consider the two possible outcomes when the agents misreport their locations:

\smallskip\noindent{\em Case where $\IntLength' \geq \IntLength$.}
Let $i$ be any agent. If $i$'s true location is covered by some interval of the new covering, $i$ incurs an expected cost of $C(\IntLength') \geq C(\IntLength)$. Otherwise, agent $i$ incurs an expected cost no less than $C(\IntLength')$, which is greater than $C(\IntLength)$.

\smallskip\noindent{\em Case where $\IntLength' < \IntLength$.}
We consider the distance of any agent to the nearest midpoint of an interval.
The locations of the truthful agents in $N \setminus S$ are covered by some interval of the new covering. Hence, their distance to the nearest midpoint of some interval is at most $\IntLength'/2$. On the other hand, if we consider the true locations of all agents and any feasible covering of them with $k$ intervals, there is some agent whose distance to the midpoint of the interval covering him is at least $\IntLength/2$. Therefore, there is an agent $i$ whose distance $d$ to the nearest midpoint of some interval in the new covering (after the deviation) is at least $\IntLength/2$. Hence, agent $i$ must be in the deviating coalition $S$, and his true location must not be covered by the intervals of the new covering. In this case, Lemma~\ref{l:ec-strproof} below implies that the expected cost of agent $i$ after the deviation, which is $\Exp[c(d-\IntLength'/2+X(\IntLength'))]$, is at least as large as $\Exp[c(X(2d))] = C(2 d) \ge C(\IntLength)$. This implies that {\sc Equal Cost} is group strategyproof.

\begin{lemma}\label{l:ec-strproof}
For all $a$, $a'$, $b$, with $0 \leq a < a' \leq b$, it holds that 
\[ \Exp[c(b-a+X(2a))] \ge \Exp[c(b-a'+X(2a'))] \] 
Moreover, the inequality is strict, if the function $c$ is strictly concave.
\end{lemma}

\begin{proof}
Let $m \geq 0$ be any integer. We only need to show that the lemma holds for all $a,a' \in [\frac{m}2,\frac{m+1}2)$, with $0 \leq a < a' \leq b$. For all such $a$, $a'$, $b$, we have that:
\begin{eqnarray*}
\Exp[c(b-a+X(2a))] &=& \sum_{i=0}^{m} p_i^m ( c(b-a+i) + c(b+a-i) ) \\
&\ge& \sum_{i=0}^{m} p_i^m ( c(b-a'+i) + c(b+a'-i) )
= \Exp[c(b-a'+X(2a'))]
\end{eqnarray*}
where the inequality holds because $a < a'$ and $c$ is concave. In fact, the inequality is strict if $c$ is strictly concave.
\qed\end{proof}

\subsection{Approximation Ratio}

In this section, we analyze the approximation ratio of {\sc Equal Cost}.

\begin{lemma}\label{l:ec-cost}
For any concave cost function $c$, any locations profile $\vec{x}$, and any agent $i$, it holds that $\cost(x_i, \EC(\vec{x}, c)) \leq \MC^\ast(\vec{x}, c)$.
\end{lemma}

\begin{proof}
We let $\IntLength$ be the minimum interval length in Step~1 of {\sc Equal Cost}, and let $m = \tfloor{\IntLength}$. We recall that $\MC^\ast(\vec{x}, c) = c(\IntLength/2)$. Moreover, we have that:
$$C(\IntLength) = \sum_{i=0}^{m} p_i^{m} ( c(i) + c(\IntLength-i) ) \le
\sum_{i=0}^{m} 2 p_i^{m} c(\IntLength/2) = c(\IntLength/2)$$
where the inequality follows from the concavity of the cost function $c$.
\qed\end{proof}

\begin{lemma}\label{l:eq-max-cost}
For every concave cost function $c$, {\sc Equal Cost} has an approximation ratio of at most $2$ for the objective of {\sc Max Cost}.
\end{lemma}

\begin{proof}
Let $(\vec{x}, c)$ be any instance with a concave cost function $c$, and let $\IntLength$ be the minimum interval length in Step~1 of {\sc Equal Cost}. In $\EC(\vec{x}, c)$, every agent $i$ has a facility at distance at most $\IntLength$ to $x_i$. On the other hand, $\MC^\ast(\vec{x}, c) = c(\IntLength/2)$. Therefore, the approximation ratio is at most:
$$\frac{c(\IntLength)}{c(\IntLength/2)} = \frac{c(\IntLength) + c(0)}{c(\IntLength/2)} \le \frac{2 c(\IntLength/2)}{c(\IntLength/2)} = 2\,,$$
where we use that $c(0) = 0$, by normalization, and the concavity of $c$.
\qed\end{proof}

\begin{lemma}\label{l:ec-social-cost}
For every concave cost function $c$, {\sc Equal Cost} has an approximation ratio of at most $n$ for the objective of {\sc Social Cost}.
\end{lemma}

\begin{proof}
For every locations profile $\vec{x}$, $\MC^\ast(\vec{x}, c) \leq \SC^\ast(\vec{x}, c)$. Then,
\[ \SC(\vec{x}, c) = \sum_{i \in N} \cost(x_i, \EC(\vec{x}, c))
 \leq n\,\MC^\ast(\vec{x}, c) \leq n\,\SC^\ast(\vec{x}, c)\,, \]
where the inequality follows from Lemma~\ref{l:ec-cost}.
\qed\end{proof}

\section{Applications}
\label{s:applications}

In this section, we consider three typical examples of concave cost functions, and derive closed form solutions for the corresponding random variables $X(\IntLength)$.

\smallskip\noindent{\em Linear Functions.}
The literature mostly focuses on linear cost functions $c(d) = \Slope d$, where the agents' cost is proportional to their distance to the nearest facility. In this case, $X(\IntLength)$ has a nice closed form: it is either $0$ with probability $1/2$ or $\IntLength$ with probability $1/2$.
%
%
Then, the expected connection cost of any location $x \in [0,\IntLength]$ is:
$$c(x)/2 + c(\IntLength-x)/2 = \Slope x /2  + \Slope (\IntLength-x)/2 = 2 \Slope \IntLength / 2\,,$$
which does not depend on $x$.

\smallskip\noindent{\em Two-Piece Piecewise Linear Functions.}
For some $\Slope_1 > \Slope_2 > 0$, let the cost function $c$ be:
$$ c(d) = \left\{
     \begin{array}{lr}
       \Slope_1 d & \text{for } d \le 1\\
       \Slope_2 d + (\Slope_1-\Slope_2) & \text{for } d > 1
     \end{array} \right.$$

To achieve the same expected cost at all locations, we find $\IntLength$, let $m = \tfloor{\IntLength}$, and compute the probability distribution of $X(\IntLength)$ by solving the following linear system:
$$\left( \begin{array}{ccccccccccc}
\Slope_1 & -\Slope_1 & -\Slope_1 & -\Slope_2 & -\Slope_2 & -\Slope_2 &  -\Slope_2 & \ldots & -\Slope_2& -\Slope_2& -\Slope_2 \\
\Slope_1 & \Slope_1 & -\Slope_1 & -\Slope_1 & -\Slope_2 & -\Slope_2 & -\Slope_2 &  -\Slope_2 & \ldots & -\Slope_2& -\Slope_2 \\
\vdots & \vdots & \vdots & \vdots & \vdots & \vdots & \vdots & \vdots & \vdots & \vdots & \vdots \\
\Slope_2 & \Slope_2 & \Slope_2 & \ldots & \ldots & \ldots & \Slope_2 & \Slope_2 & \Slope_1 & \Slope_1 & -\Slope_1
 \end{array} \right) \left( \begin{array}{c}
p^m_0 \\
p^m_m \\
p^m_1 \\
p^m_{m-1} \\
\vdots\\
p^m_m \\
p^m_0
 \end{array} \right) = 0$$
Taking the difference between every two consecutive rows, as in Lemma \ref{l:random-var}, we find that:
$$p^m_i = \frac {\Slope_1-\Slope_2} {2 \Slope_1} (p^m_{i-1} + p^m_{i+1})\ \ \
\text{ for all integers $i$ }\, 0 \le i \le m\,,$$
where we define $p^m_i = 0$, for all integers $i \not \in [0,m]$.
Then, the solution of the recurrence is:
$$p^m_i = \frac{ \rho_1^{m+1-i} + \rho_2^{m+1-i} } { 2 \sum_{j=1}^{m+1} \left( \rho_1^{j} + \rho_2^{j} \right) }$$
$$\text{ where\ \ \ } \rho_1 = \frac {\Slope_1+\sqrt{\Slope_1^2 - (\Slope_1-\Slope_2)^2}}{\Slope_1-\Slope_2} \text{\ \ \ and\ \ \ } \rho_2 = \frac {\Slope_1-\sqrt{\Slope_1^2 - (\Slope_1-\Slope_2)^2}}{\Slope_1-\Slope_2}$$

\noindent{\em Exponential Functions.}
A concave cost function that results in a continuous probability distribution $X(\IntLength)$ is the exponential function $c(d) = 1-e^{-\lambda d}$.
Then, $X(\IntLength)$ is $0$, with probability $\frac{1}{\IntLength \lambda + 2}$, $\IntLength$, with probability $\frac{1}{\IntLength \lambda + 2}$, and
uniform in $(0,\IntLength)$, with probability $\frac{\IntLength \lambda}{\IntLength \lambda + 2}$.

We let $X(\IntLength)$ be $0$, with probability $\frac{1}{\IntLength \lambda + 2}$, $\IntLength$, with probability $\frac{1}{\IntLength \lambda + 2}$, and
uniform in $(0,\IntLength)$, with probability $\frac{\IntLength \lambda}{\IntLength \lambda + 2}$. We next show that the expected connection cost of any location $x \in [0,\IntLength]$ does not depend on $x$. In particular, the expected connection cost of any location $x$ is:
$$\frac{1}{\IntLength \lambda + 2} c(x) + \frac{1}{\IntLength \lambda + 2} c(\IntLength-x) +
\frac{\IntLength \lambda}{\IntLength \lambda + 2} \int_0^{\IntLength} \frac 1 {\IntLength} c(|t-x|) \mathrm{d}t=$$
$$\frac{c(x)+c(\IntLength-x)}{\IntLength \lambda + 2}  +
\frac{ \lambda}{\IntLength \lambda + 2} \int_0^x c(x-t) \mathrm{\IntLength}t+
\frac{ \lambda}{\IntLength \lambda + 2} \int_x^\IntLength c(t-x) \mathrm{d}t=$$
$$\frac{2-e^{-\lambda x}-e^{-\lambda (\IntLength-x)}}{\IntLength \lambda + 2}  +
\frac{ \lambda}{\IntLength \lambda + 2}
\left(\int_0^x  1-e^{-\lambda (x-t)} \mathrm{d}t+
 \int_x^\IntLength 1-e^{-\lambda (t-x)} \mathrm{d}t \right)=$$
$$\frac{2-e^{-\lambda x}-e^{-\lambda (\IntLength-x)}}{\IntLength \lambda + 2}  +
\frac{ \lambda}{\IntLength \lambda + 2}
\left(\IntLength - \int_0^x  e^{-\lambda (x-t)} \mathrm{d}t+
 \int_x^\IntLength e^{-\lambda (t-x)} \mathrm{d}t \right)=$$
$$\frac{2-e^{-\lambda x}-e^{-\lambda (\IntLength-x)}}{\IntLength \lambda + 2}  +
\frac{ \lambda}{\IntLength \lambda + 2}
\left(\IntLength -
 \frac{ 1- e^{-\lambda x} } {\lambda} - \frac { 1- e^{-\lambda (\IntLength-x)} } {\lambda} \right)=\frac{ \IntLength \lambda}{\IntLength \lambda + 2} $$
which does not depend on $x$.

\section{Extensions and Limitations}
\label{s:limitations}

\subsection{{\sc Equal Cost} in Bounded Intervals}

Our results about the properties of {\sc Equal Cost} apply to the real line $(-\infty,\infty)$ and to the half-line $[0,\infty)$. If the metric space is a bounded interval $[0, L]$, it could be that in the construction of the covering, in Step~1, the last interval does not fit entirely in $[0, L]$. The following lemma shows that even in this case, we can adjust the covering with disjoint intervals of the same length, computed in Step~1, so that all intervals fit in $[0, L]$.

\begin{lemma}\label{l:bounded_interval}
Given a locations profile $\vec{x}$ in $[0,L]$, there is an optimal covering of $\vec{x}$ with $k$ disjoint intervals of the same (minimum) length that all lie entirely in $[0,L]$.
\end{lemma}

\begin{proof}
We consider a covering of $\vec{x}$ with $k$ disjoint intervals of the same minimum length $\IntLength$, computed as in Section~\ref{s:partitioning}. As in Step~1 of {\sc Equal Cost}, we number the intervals from left to right, and let the $i$-th interval be $[\IntStart_i,\IntStart_i + \IntLength]$.
Since all the locations of $\vec{x}$ lie in $[0,L]$, we obtain that $\IntLength \leq L/k$. Moreover, by construction, we have that $\IntStart_i \geq 0$, for all $1 \leq i \leq k$. However, it could be $\IntStart_i + \IntLength > L$ for some interval $i$.
In this case, we construct a new covering using the intervals $[\IntStart'_i,\IntStart'_i + \IntLength]$, $i = 1, \ldots, k$, where $\IntStart'_i = \min \{ \IntStart_i, L-(k+1-i)\IntLength \}$.
To show that this is indeed an admissible covering of $\vec{x}$, we observe that:

\smallskip\noindent(i) All intervals lie entirely in $[0,L]$: For every $i$, $\IntStart'_i \geq 0$, since $\IntStart_i \geq 0$, and $L-(k+1-i)\IntLength
\geq 0$, because $\IntLength \leq L/k$. Furthermore,
$\IntStart'_i + \IntLength \leq L-(k+1-i)\IntLength + \IntLength \leq L-(k-i)\IntLength \leq L$.

\smallskip\noindent(ii) All intervals are disjoint: For any two consecutive intervals $i$ and $i+1$, we have that:
\begin{align*}
\IntStart'_{i+1}-\IntStart'_i &=
\min\{ \IntStart_{i+1}, L-(k+1-i-1)\IntLength \} -
\min\{ \IntStart_{i}, L-(k+1-i)\IntLength \} \\
&\geq \min\{ \IntStart_{i} + \IntLength, L-(k+1-i-1)\IntLength \} -
      \min\{ \IntStart_{i}, L-(k+1-i)\IntLength \}\\
&= \IntLength + \min\{ \IntStart_{i}, L-(k+1-i)\IntLength \} -
                \min\{ \IntStart_{i}, L-(k+1-i)\IntLength \}\\
&= \IntLength
\end{align*}

\noindent(iii) The intervals cover all locations of $\vec{x}$: Let us consider a location $x \in [\IntStart_i,\IntStart_i+\IntLength]$. If $\IntStart'_i = \IntStart_i$, $x$ is covered since the interval does not change.
Otherwise, $\IntStart'_i = L -(k+1-i)\IntLength$. Thus, the interval $[\IntStart'_i,L]$ has a length of $L - \IntStart'_i = (k+1-i)\IntLength$, and consists of $k+1-i$ disjoint intervals of length $\IntLength$. Therefore, the intervals $[\IntStart'_j, \IntStart'_j+\IntLength]$, for $j \geq i$, entirely cover the interval $[\IntStart'_i, L] \supset [\IntStart_i, L]$, and thus, they also cover the location $x$.
\qed\end{proof}

Lemma~\ref{l:bounded_interval} implies that if the agents lie on a circle, we can also cover their locations with disjoint intervals of the same minimum length $\IntLength$. Then, we can apply Steps~2~and~3 to the resulting intervals on the circle. But rather surprisingly, {\sc Equal Cost} is guaranteed to be strategyproof for $k$-Facility Location on the circle only if $k$ is even. Otherwise, some agents in the first interval may prefer the facility placed in the last interval, which violates the property that each agent always prefers the facility in his own interval.

\subsection{Convex Cost Functions}
\label{s:convex}

The approach of {\sc Equal Cost} does not apply to strictly convex functions $c$, because it is no longer possible to equalize the expected cost of all agents. To see this, let us consider the interval $[0, \IntLength]$, and the expected cost of two agents, one located at $0$ and the other at $\IntLength$. Since $\Exp[c(X)] + \Exp[c(\IntLength-X)] > \Exp[2 c(\IntLength/2)] = 2 c(\IntLength/2)$, by the strict convexity of $c$, at least one of them incurs an expected cost greater than $c(\IntLength/2)$. However, a third agent located at $\IntLength/2$ incurs an expected cost no greater than $c(\IntLength/2)$, since his distance to the facility is at most $\IntLength/2$. Moreover, we can show that:

\begin{lemma}\label{l:convex}
There is no randomized strategyproof mechanism that achieves a bounded approximation ratio for the class of all convex functions.
\end{lemma}

\begin{proof}
We recall that the property of a bounded approximation ratio is objective-independent. So, we next focus on the objective of {\sc Max Cost}. For the proof, we consider the convex cost function $c(d) = e^d$ and instances with $2$ agents and a single facility.
For sake of contradiction, we assume that there exists a randomized strategyproof mechanism that achieves an approximation ratio of $r$ for such instances. Next, we let $X$ denote the random variable that determines where the mechanism places the facility.

We first consider an instance $\vec{x} = (x_1, x_2)$, with $x_2 > x_1$, If the facility is placed at location $t \leq (x_1 + x_2)/2$, agent $2$ incurs the maximum cost equal to $e^{x_2-t}$. If the facility is placed at $t > (x_1 + x_2)/2$, agent $1$ incurs the maximum cost equal to $e^{t-x_1}$. In both cases, the maximum cost is equal to $e^{(x_2-x_1)/2+|t-(x_1+x_2)/2|}$, and the expectation of the maximum cost is $\Exp[e^{(x_2-x_1)/2+|X-(x_1+x_2)/2|}] \leq r e^{(x_2-x_1)/2}$, which implies that $\Exp[e^{|X-(x_1+x_2)/2|}] \leq r$.

Let us now consider the probabilities $p_l = \Prob[X \leq \frac{x_1 + x_2}{2}]$ and $p_r = \Prob[X \geq \frac{x_1 + x_2}{2}]$. Since $p_l + p_r \geq 1$, one of them is at least $1/2$. Wlog., let us assume that $p_l \geq 1/2$, which implies that agent 2 incurs an expected cost of at least $\frac{1}{2} e^{(x_2-x_1)/2}$.

Next, we consider an instance $\vec{x}' = (x'_1, x'_2)$, with $x'_1 = x_1$ and $x'_2 = 2 x_2 - x_1$. By the choice of $\vec{x}'$,
$\Exp[e^{|X-(x'_1+x'_2)/2|}] = \Exp[e^{|X-(x_1+2 x_2 - x_1)/2|}] = \Exp[e^{|X-x_2|}]$.
Working as before, we obtain that $\Exp[e^{|X-(x'_1+x'_2)/2|}] = \Exp[e^{|X-x_2|}] \leq r$, due to the approximation ratio of the mechanism. Moreover, $\Exp[e^{|X-x_2|}]$ is the expected cost of an agent located at $x_2$, and due to strategyproofness, is no less than the expected cost of agent $2$ in instance $\vec{x}$. Otherwise agent $2$ would have an incentive to report $x'_2$, instead of $x_2$. Therefore, $\Exp[e^{|X-x_2|}] \geq \frac{1}{2} e^{(x_2-x_1)/2}$.  Combining the upper and the lower bound on $\Exp[e^{|X-x_2|}]$, we obtain that $ e^{(x_2-x_1)/2} \leq 2r$. This leads to a contradiction if we consider an instance $\vec{x}$ with $x_2-x_1 > 2 \ln(2r)$.
\qed\end{proof}

\subsection{Other Cost Functions}

{\sc Equal Cost} can also apply to some other (non-convex) cost functions, for which the expected cost of all agents can be equalized. A notable such example is a cost function $c_r(d)$ which is $0$, if $d < r$, and $1$ otherwise. Thus, $c_r$ correspond to agents that only care about getting a facility within a radius $r$ from their location. In this case, one could apply {\sc Equal Cost} as follows:
First, we find a covering of the agent locations with intervals of length $\IntLength$, as in Step~1. Then if $\IntLength \le 2r$, we place a facility at the midpoint of each interval. Otherwise, we do not place any facilities (and let each agent incur a cost of $1$). This clearly satisfies the equal cost property since the cost incurred by all agents is either $0$ or $1$. The mechanism is optimal for the objective of {\sc Max Cost} because every agent incurs a cost of $0$, if the optimal solution satisfies all agents, and a cost of $1$, otherwise. On the other hand, the mechanism is $n$ approximate for the objective of {\sc Social Cost}, since in case where the optimal solution satisfies all but one agents, resulting in a social cost of $1$, the mechanism does not place any facilities, and incurs a social cost of $n$.

\section{The {\sc Pick the Loser} Mechanism}
\label{s:loser}

{\sc Equal Cost} performs well for the objective of {\sc Max Cost}, but may perform poorly for the objective of {\sc Social Cost}. An extreme case is when we have $k$ facilities and only $n = k+1$ agents. Then, there are many facilities, and one could easily satisfy all but one agents. Nevertheless, {\sc Equal Cost} causes all agents to incur a high cost (equal to the min-max cost for linear cost functions).

In certain cases, this might not be acceptable, and one needs to find a more efficient mechanism. In this section, we present a mechanism that, for instances with only $n=k+1$ agents, selects the loser, i.e., the agent not allocated a facility at his location, in a group strategyproof way. We also show that this mechanism is quite efficient for the {\sc Social Cost} objective, which for such instances, is equal to the cost of the loser.

Given an instance $(\vec{x}, c)$ of $k$-Facility Location on the line with only $n = k+1$ agents, the {\sc Pick-the-Loser} mechanism, of $\PL$ in short, works as follows:
\begin{description}
\item[Step 1] It numbers the agents according to their reported locations such that $x_i < x_{i+1}$, and lets $E$ and $O$ be the sets of even and odd numbered agents, respectively. For every odd-numbered agent $i \in O$, $\PL$ places a facility at $x_i$.

\item[Step 2] For each even numbered agent $i$, $\PL$ samples a number $s_i$ uniformly in $(0,1)$, and computes $i$'s current cost $\kappa_i = \min_{j \neq i} c(|x_j-x_i|)$ and $i$'s scaled cost $\hat \kappa_i = \kappa_i/s_i$.

\item[Step 3] $\PL$ finds the agent with the smallest scaled cost, and declares him the \emph{loser}. Then, $\PL$ places facilities at the locations of all other agents.
\end{description}

In the following, we first show that {\sc Pick the Loser} is strategyproof (Lemma~\ref{l:pal-strategyproofness}). Then, in Section~\ref{s:ptl-groupstrproof}, we use strategyproofness, and deal with the case where a coalition of agents may deviate, thus establishing that the mechanism is strongly group strategyproof. Finally, in Section~\ref{s:ptl-approx}, we prove the mechanism's approximation guarantee. Thus, we obtain:

\begin{theorem}\label{th:ptl}
For the class of all concave cost functions, {\sc Pick the Loser} is strongly group strategyproof and achieves an approximation ratio of $2$ for the {\sc Social Cost} objective.
\end{theorem}

For the proof, we assume wlog. that the agent locations are all distinct. Otherwise, we allocate a facility to all distinct locations, thus being trivially both optimal and group strategyproof. We let $q_i(\vec{x})$ denote the probability that agent $i$ is designated as a loser. We have that $q_i(\vec{x}) = 0$ for all odd numbered agents in $\vec{x}$. For an even numbered agent $i$, we can compute this probability by the following thought experiment: With all the samples $s_j \in (0,1)$ fixed, agent $i$ is selected if for all $j \in E$, $\hat \kappa_j \ge \hat \kappa_i$, or equivalently if  $s_j \le s_i \kappa_j / \kappa_i$. This happens with probability $\prod_{j \in E \setminus \{ i \}} \min \{1, s_i \kappa_j / \kappa_i \}$. Setting $t = s_i / \kappa_i$ and taking the expectation over all different values of $t$, we
have that
\[ q_i(\vec{x}) = \kappa_i \int_0^{1/\kappa_i} \prod_{j \in E \setminus \{ i \}} \min \{1, \kappa_j t \} dt \]

\subsection{Strategyproofness}

The following lemma implies that {\sc Pick the Loser} is strategyproof for the class of all concave cost functions. Next, in Section~\ref{s:ptl-groupstrproof}, we use this property, to establish that {\sc Pick the Loser} is strongly group strategyproof for the class of all concave cost functions.

\begin{lemma}\label{l:pal-strategyproofness}
Let $(\vec{x}, c)$ be any instance with a concave cost function $c$ and only $n = k+1$ agents occupying $n$ distinct locations. Then, for every agent $i$ and every location $x'_i \neq x_i$, $\cost(x_i, \PL(\vec{x}, c)) < \cost(x_i, \PL((\vec{x}_{-i}, x'_i), c))$.
\end{lemma}

\begin{proof}
For convenience, we let $\vec{x}' = (\vec{x}_{-i}, x'_i)$. We also recall that by normalizing $c$, we assume that $c(0) = 0$. If agent $i$ is an odd numbered agent, he strictly prefers $\vec{x}$ over $\vec{x}'$, because in $\vec{x}$, there is a facility at $x_i$ and agent $i$ incurs $0$ cost, while in $x'_i$, there is no facility at $x_i$, and thus agent $i$ incurs a positive cost.

If $i$ is an even numbered agent, we let $\delta = \min_{j \neq i} \{|x_i-x_j|\}$ and $\delta' = \min_{j \neq i} \{|x'_i-x_j|\}$ denote the minimum distance of the reported location of $i$ to the location of another agent. In the instance $\vec{x}$, if agent $i$ is not allocated a facility at $x_i$, he incurs a cost of $c(\delta)$. Otherwise, agent $i$ incurs $0$ cost. Since $\delta > 0$ and the cost function in increasing $c$, we have that $c(\delta) > 0$.
We next consider three different cases, and show that in each case, agent $i$ prefers $\vec{x}$ to $\vec{x}'$.

\smallskip\noindent\emph{Case where $x'_i \not\in (x_i-\delta,x_i+\delta)$.}
Then, in $\vec{x}'$, agent $i$ incurs an expected cost of at least $c(\delta)$, while in $\vec{x}$, he incurs an expected cost less than $c(\delta)$, since he is allocated a facility at $x_i$ with positive probability.

\smallskip\noindent\emph{Case where $x'_i \in (x_i-\delta,x_i+\delta)$ and $\delta' \leq \delta$.}
In this case, the probability $q_i(\vec{x}')$ that agent $i$ is not allocated a facility at $x'_i$, in instance $\vec{x}'$, is greater than or equal to $q_i(\vec{x})$. This holds because $i$'s cost in $\vec{x}'$, which is $\kappa'_i = c(\delta')$, is less than or equal to $i$'s cost in $\vec{x}$, which is $\kappa_i = c(\delta)$. Therefore, for any sampled number $s_i$, agent $i$ has a smaller scaled cost $\hat \kappa'_i$ in instance $\vec{x}'$ than his corresponding scaled cost $\hat \kappa_i$ in instance $\vec{x}$, which in turn, implies a greater probability that $i$ is designated as the loser.
Moreover, if in instance $\vec{x}'$, agent $i$ is allocated a facility at $x'_i$, he incurs a positive cost, since $x'_i$ is different from his true location $x_i$. Thus, putting everything together, we obtain that agent $i$ strictly prefers $\vec{x}$ to $\vec{x}'$:
$$\left(1-q_i(\vec{x}')\right) c(|x'_i-x_i|) + q_i(\vec{x}') c(\delta) >
q_i(\vec{x}') c(\delta) \geq q_i(\vec{x}) c(\delta)$$

\smallskip\noindent\emph{Case where $x'_i \in (x_i-\delta,x_i+\delta)$ and $\delta' > \delta$.}
The probability $q_i(\vec{x}')$ is now greater than the probability $q_i(\vec{x})$. However, if agent $i$ is allocated a facility at $x'_i$, in instance $\vec{x}'$, he incurs an additional cost of $c(|x'_i-x_i|) \geq c(\delta'-\delta)$, due to the distance of $x'_i$ to $i$'s true location $x_i$. Thus, we obtain the following lower bound on the expected cost of agent $i$ in instance $\vec{x}'$:
$$\left(1-q_i(\vec{x}')\right) c(|x'_i-x_i|) + q_i(\vec{x}') c(\delta) \geq
\left(1-q_i(\vec{x}')\right) (c(\delta')-c(\delta)) + q_i(\vec{x}') c(\delta)\,,$$
where the inequality follows from $c(\delta') \leq c(\delta) + c(\delta'-\delta)$, which in turn, follows from the concavity of $c$.
Hence, to conclude that agent $i$ strictly prefers $\vec{x}$ to $\vec{x}'$, we need to show that:
\begin{equation}\label{eq:ptl-strproof}
\left(1-q_i(\vec{x}')\right) (c(\delta')-c(\delta)) + q_i(\vec{x}') c(\delta) >
q_i(\vec{x}) c(\delta)
\end{equation}

To this end, for each even numbered agent $j$, we let $\kappa_j$ and $\kappa'_j$ denote the cost of $j$ computed by the mechanism for the instances $\vec{x}$ and $\vec{x}'$, respectively. By the definition of the mechanism, we have that $\kappa_j = \kappa'_j$, for any agent $j \in E \setminus \{i \}$, and that $\kappa_i = c(\delta)$ and $\kappa'_i = c(\delta')$.
Hence, the probability $q_i(\vec{x}')$ can be calculated as follows:
\begin{equation}\label{eq:qi}
q_i(\vec{x}') = \kappa'_i \int_0^{1/\kappa'_i} \prod_{j \in E \setminus \{ i \}} \min \{1, \kappa_j t \} dt
\end{equation}

To prove (\ref{eq:ptl-strproof}), we show that $i$' expected cost is increasing with $\kappa'_i$. To prove this, we show that the partial derivative of $i$'s cost with respect to $\kappa'_i$ is positive. Formally, we show that:
\begin{equation}\label{eq:ptl-deriv}
\frac {\partial} {\partial \kappa'_i} [ \left(1-q_i(\vec{x}')\right) (\kappa'_i-\kappa_i) + q_i(\vec{x}') \kappa_i ] > 0
\end{equation}

We first substitute $q_i(\vec{x}')$, with the use of (\ref{eq:qi}), and the left-hand-side of (\ref{eq:ptl-deriv}) becomes:
\[
\frac {\partial} {\partial \kappa'_i} \left[ \left(1- \kappa'_i \int\limits_0^{1/\kappa'_i}\! \prod_{j \in E \setminus \{ i \}}\!\!\!\!
\min \{1, \kappa_j t \} dt\right) (\kappa'_i-\kappa_i) + \kappa_i \left( \kappa'_i \int\limits_0^{1/\kappa'_i}\! \prod_{j \in E \setminus \{ i \}}\!\!\!\! \min \{1, \kappa_j t \} dt \right) \right]
\]
Next, we calculate the partial derivative with respect to $\kappa'_i$, and the quantity above becomes:
\[
1-\!\!\!\!\prod_{j \in E \setminus \{ i \}}\!\!\!\! \min \left\{1, \frac {\kappa_j} {\kappa'_i} \right\} +
\frac {2(\kappa'_i - \kappa_i)} {\kappa'_i} \left(  \prod_{j \in E \setminus \{ i \}}\!\!\!\! \min \left\{1, \frac {\kappa_j} {\kappa'_i} \right\} - \kappa'_i \int\limits_0^{1/\kappa'_i}\!\! \prod_{j \in E \setminus \{ i \}}\!\!\!\! \min \{1, \kappa_j t \} dt \right)
\]
Using that $\prod_{j \in E \setminus \{ i \}} \min \{1, \kappa_j t \} \leq \prod_{j \in E \setminus \{ i \}} \min \{1,  \kappa_j / \kappa'_i \}$, which holds for all $t \in [0, 1 / \kappa'_i ]$, and with strict inequality for $t < 1/\kappa'_i$, we obtain that the quantity above is greater than:
\[
1-\!\!\!\!\prod_{j \in E \setminus \{ i \}}\!\!\!\! \min \left\{1, \frac {\kappa_j} {\kappa'_i} \right\} +
\frac {2(\kappa'_i - \kappa_i)} {\kappa'_i} \left(  \prod_{j \in E \setminus \{ i \}}\!\!\!\! \min \left\{1, \frac {\kappa_j} {\kappa'_i} \right\} - \kappa'_i  \prod_{j \in E \setminus \{ i \}}\!\!\!\! \min \left\{1, \frac{\kappa_j}{\kappa'_i} \right\} \int\limits_0^{1/\kappa'_i} 1 dt \right)
\]
Simplifying the quantity above and returning back to (\ref{eq:ptl-deriv}), we conclude that:
\[
\frac {\partial} {\partial \kappa'_i} [ \left(1-q_i(\vec{x}')\right) (\kappa'_i-\kappa_i) + q_i(\vec{x}') \kappa_i ] >
1-\!\!\prod_{j \in E \setminus \{ i \}}\!\!\!\min \{1, \kappa_j / \kappa'_i \} \geq 0
\]

Therefore, the expected cost of agent $i$ is increasing with $\kappa'_i$. Hence, we obtain that 
\[ \left(1-q_i(\vec{x}')\right) (\kappa'_i-\kappa_i) + q_i(\vec{x}') \kappa_i > q_i(\vec{x}) \kappa_i\,, \]
which is identical to (\ref{eq:ptl-strproof}). This proves that in the third case, agent $i$ strictly prefers $\vec{x}$ to $\vec{x}'$, and concludes the proof of the lemma.
\qed\end{proof}

\subsection{Strong Group Strategyproofness}
\label{s:ptl-groupstrproof}

Proving that {\sc Pick the Loser} is strong group strategyproof requires some additional arguments and case analysis, where we use that the mechanism is strategyproof (Lemma~\ref{l:pal-strategyproofness}).

Throughout this section, we consider an instance $\vec{x}$ with $n$ distinct locations, where the agents are numbered as they appear on the line, from left to right. Hence, we have that $x_i < x_{i+1}$, for all $i = 1, \ldots, n-1$. We prove that there is no coalition of agents that can benefit by misreporting their location. For sake of contradiction, let us assume that such a coalition exists. In particular, we let $S$ be such coalition of minimum size, and let $\vec{x}' = (\vec{x}'_S, \vec{x}_{-S})$ be the new instance, where the agents in $S$ misreport their location. By the definition of strong group strategyproofness, for all $i \in S$, $ \cost(x_i, \PL(\vec{x}', c)) \le \cost(x_i, \PL(\vec{x}, c))$, and the inequality is strict for at least one agent in $S$.

We observe that for every odd numbered agent $i$, $x_i \in \vec{x}'$. Otherwise, agent $i$ would incur a positive cost in $\vec{x}'$, and would prefer $\vec{x}$ to $\vec{x}'$. Since {\sc Pick the Loser} is anonymous, i.e., does not take the agent identities into account, we can assume wlog. that $x'_i = x_i$, which implies that the deviating coalition $S$ doesn't contain any odd numbered agents.

Furthermore, we observe that for every even numbered agent $i$, there is a location in $\vec{x}'$ lying in the interval $(x_{i-1}, x_{i+1})$, where $x_{n+1}$ is defined to be $\infty$. Otherwise agent $i$ would incur an expected cost of $\cost(x_i, \PL(\vec{x}', c)) \geq \min\{c(x_i-x_{i-1}), c(x_{i+1}-x_i)\}$, which is greater than his expected cost for instance $\vec{x}$, where $x_i$ is allocated a facility with positive probability. Again, since {\sc Pick the Loser} is anonymous, we can assume wlog. that $x'_i \in (x_{i-1}, x_{i+1})$, which implies that the relative order of the agents in $\vec{x}'$ is the same as in $\vec{x}$.

Let us now consider an agent $i \in S$, and let $\kappa_i$ and $\kappa'_i$ denote the cost of $i$ computed by the mechanism for the instances $\vec{x}$ and $\vec{x}'$, respectively. Next, we exclude the possibility that $\kappa'_i > \kappa_i$. Specifically, we show that if $\kappa'_i > \kappa_i$, the instance $\vec{x}'' = (\vec{x}'_{-i}, x_i)$ is strictly preferable to $\vec{x}'$ for all agents in $S$. That holds because, in $\vec{x}''$, agent $i$ has cost $\kappa_i < \kappa'_i$, and therefore, the probability that he is designated as the loser in $\vec{x}''$ is greater than the corresponding probability in $\vec{x}'$. Hence, for every agent $j \in E \setminus \{i\}$, the probability that agent $j$ is designated as the loser in $\vec{x}''$ is less than the corresponding probability in $\vec{x}'$, which implies that agent $j$ strictly prefers the instance $\vec{x}''$ to the instance $\vec{x}'$. Also by Lemma~\ref{l:pal-strategyproofness}, {\sc Pick the Loser} is strategyproof, and thus, agent $i$ strictly prefers the instance $\vec{x}''$ to the instance $\vec{x}'$. However, since the number of agents misreporting their locations in $\vec{x}''$ is one less than the corresponding number in $\vec{x}'$, this contradicts the hypothesis that $S$ is the smallest coalition of agents that can benefit from misreporting their location.

So, let us now assume that $\kappa'_i < \kappa_i$, for all agents $i \in S$, and let $\rho = \min\{ \kappa_i / \kappa'_i \} > 1$. We consider an instance $\vec{x}''$
where the cost $\kappa''_i$ computed by the mechanism for all agents $i \in S$ is equal to $\rho \kappa'_i$. Such an instance $\vec{x}''$ can be obtained if we let all agents $i \in S$ report locations closer to their original location. We next prove that for every agent $i \in S$, the probability $q_i(\vec{x}'')$ that $i$ is designated as the loser in $\vec{x}''$ is less than the probability $q_i(\vec{x}')$ that $i$ is designated as the loser in $\vec{x}'$. More precisely:
\begin{align*}
q_i(\vec{x}'') &= \Prob[ \hat \kappa''_i < \hat \kappa''_j,\forall j \not \in S\,|\,\hat \kappa''_i < \hat \kappa''_j, \forall j \in S \setminus \{i\}] \cdot \Prob[ \hat \kappa''_i < \hat \kappa''_j, \forall j \in S \setminus \{i\}]\\
&= \Prob[ \rho \hat \kappa'_i < \hat \kappa'_j,\forall j \not \in S\,|\,\rho \hat \kappa'_i < \rho \hat \kappa'_j, \forall j \in S \setminus \{i\}] \cdot \Prob[ \rho \hat \kappa'_i < \rho \hat \kappa'_j, \forall j \in S \setminus \{i\}]\\
&= \Prob[ \hat \kappa'_i < \hat \kappa'_j / \rho,\forall j \not \in S\,|\,\hat \kappa'_i < \hat \kappa'_j, \forall j \in S \setminus \{i\}] \cdot \Prob[ \hat \kappa'_i < \hat \kappa'_j, \forall j \in S \setminus \{i\}]\\
&< \Prob[ \hat \kappa'_i < \hat \kappa'_j, \forall j \not \in S\,|\,\hat \kappa'_i < \hat \kappa'_j, \forall j \in S \setminus \{i\}] \cdot \Prob[ \hat \kappa'_i < \hat \kappa'_j, \forall j \in S \setminus \{i\}]\\
& = q_i(\vec{x}')
\end{align*}

Since for every agent $i \in S$, (i) the probability that $i$ is designated as the loser is smaller in $\vec{x}''$ than in $\vec{x}$, i.e., $q_i(\vec{x}'') < q_i(\vec{x}')$, (ii) the reported location of $i$ in $\vec{x}''$ is closer to his true location $x_i$ than his reported location in $\vec{x}'$, i.e., $|x''_i - x_i| < |x'_i - x_i|$, and (iii) there are no odd numbered agents in $S$, all agents $i \in S$ strictly prefer $\vec{x}''$ to $\vec{x}'$.

Therefore, we can assume that in the instance $\vec{x}'$, there is an agent $i \in S$ with $\kappa'_i = \kappa_i$. We now consider the instance $\vec{x}'' = (\vec{x}'_{-i}, x_i)$, where the agent $i$ is removed from the deviating coalition $S$. We note that for every agent $j \in S$, the probability that agent $j$ is designated as the loser in $\vec{x}''$ is the same as the corresponding probability in $\vec{x}'$,
i.e., $q_j(\vec{x}'') = q_j(\vec{x}')$. Therefore, the expected cost of every agent $j \in S \setminus \{ i \}$ in $\vec{x}''$ is the same as his expected cost in $\vec{x}'$. Moreover, by Lemma~\ref{l:pal-strategyproofness}, {\sc Pick the Loser} is strategyproof, and thus, the expected cost of agent $i$ in $\vec{x}''$ is less than his expected cost in $\vec{x}'$. Therefore, if the agents in the coalition $S$ can benefit by misreporting their locations, the same holds for the $S \setminus \{ i\}$. However, this contradicts the hypothesis that $S$ is the smallest coalition of agents that can benefit from misreporting their location. Hence, we have shown that such a coalition $S$ does not exist, and thus, the mechanism {\sc Pick the Loser} is strongly group strategyproof.

\subsection{Approximation Ratio}
\label{s:ptl-approx}

\begin{lemma}\label{l:ptl-approx}
For all concave cost functions, {\sc Pick the Loser} achieves an approximation ratio of at most $2$ for the objective of {\sc Social Cost}, and an approximation ratio of at most $4$ for the objective of {\sc Max Cost}.
\end{lemma}

\begin{proof}
Let $(\vec{x}, c)$ be any instance with concave $c$, and let $q$ be an agent with $\kappa_{q} = \min_i\{ \kappa_i \}$. Then, $\SC^\ast(\vec{x}, c) = \kappa_q$, while the {\sc Social Cost} of the mechanism is equal to:
\begin{align*}
\sum_i \kappa_i  \int_0^{1/\kappa_i} \kappa_i \prod_{j \in E \setminus \{ i \}} \min \{1, \kappa_j t \} dt 
& \le \kappa_q + \sum_{i \neq q} \kappa_i  \int_0^{1/\kappa_i} \kappa_i \prod_{j \in E \setminus \{ i \}} \min \{1, \kappa_j t \} dt\\
&= \kappa_q + \sum_{i \neq q} \kappa_i  \int_0^{1/\kappa_i} \kappa_i \kappa_q t \prod_{j \in E \setminus \{ i, q \}} \min \{1, \kappa_j t \} dt \\ 
& \le \kappa_q + \sum_{i \neq q} \kappa_i  \int_0^{1/\kappa_i} \kappa_q \prod_{j \in E \setminus \{ i, q \}} \min \{1, \kappa_j t \} dt \\
& \le \kappa_q + \kappa_q \sum_{i \neq q} \int_0^{1/\kappa_i} \kappa_i \prod_{j \in E \setminus \{ i, q \}} \min \{1, \kappa_j t \} dt \\
&= \kappa_q + \kappa_q \sum_{i \neq q} \Prob[\hat \kappa_i < \hat \kappa_j, \forall j \not\in \{i, q\}] dt \\
& = \kappa_q + \kappa_q \cdot 1 = 2 \kappa_q
\end{align*}

Moreover, since $c$ is concave,
$\MC^\ast(\vec{x}, c) \geq \kappa_q / 2 = 2\kappa_q / 4 \geq \MC(\PL(\vec{x}, c))/4$. \qed\end{proof}

\section{Open Problems}
\label{s:concl}

There are a few interesting open problems arising from our work. First, since {\sc Equal Cost} crucially depends on the linear structure of the instances, it would be interesting to have a mechanism that can be applied to more general metric spaces, and retains the nice properties of {\sc Equal Cost}.
Another intriguing open problem has to do with the approximability of {\sc Social Cost} by randomized strategyproof mechanisms. Despite the considerable interest in the problem, we do not know whether there exists a randomized mechanism for $k$-Facility Location that achieves an approximation ratio of $o(n)$ for all $k \geq 3$.
Another, more general, direction for further research may concern the role of the cost function $c$, which we assume here to be the same for all players. It would be interesting to investigate the approximability of $k$-Facility Location on the line if each agent $i$ may have a different concave cost function $c_i(d)$. A good starting point in this direction may be a simple setting where each agent $i$ is associated with a tuple $(x_i, r_i)$, with possibly both $x_i$ and $r_i$ being private information, and there is some fixed small cost incurred by agent $i$, if there is a facility within a distance of $r_i$ to $x_i$, and some fixed large cost incurred by agent $i$, otherwise.


\begin{thebibliography}{}

\bibitem[\protect\citeauthoryear{Alon, Feldman, Procaccia, and
  Tennenholtz}{Alon et~al\mbox{.}}{2010}]{AFPT09}
{\sc Alon, N.}, {\sc Feldman, M.}, {\sc Procaccia, A.}, {\sc and} {\sc
  Tennenholtz, M.} 2010.
\newblock Strategyproof approximation of the minimax on networks.
\newblock {\em Mathematics of Operations Research\/}~{\em 35,\/}~3, 513--526.

\bibitem[\protect\citeauthoryear{Barber\`{a}}{Barber\`{a}}{2001}]{Bar01}
{\sc Barber\`{a}, S.} 2001.
\newblock An introduction to strategyproof social choice functions.
\newblock {\em Social Choice and Welfare\/}~{\em 18}, 619--653.

\bibitem[\protect\citeauthoryear{Dokow, Feldman, Meir, and Nehama}{Dokow
  et~al\mbox{.}}{2012}]{DFMN12}
{\sc Dokow, E.}, {\sc Feldman, M.}, {\sc Meir, R.}, {\sc and} {\sc Nehama, I.}
  2012.
\newblock Mechanism design on discrete lines and cycles.
\newblock In {\em Proc. of the 13th {ACM} Conference on Electronic Commerce
  (EC~'12)}. 423--440.

\bibitem[\protect\citeauthoryear{Escoffier, Gourv\`{e}s, Thang, Pascual, and
  Spanjaard}{Escoffier et~al\mbox{.}}{2011}]{EGTPS11}
{\sc Escoffier, B.}, {\sc Gourv\`{e}s, L.}, {\sc Thang, N.}, {\sc Pascual, F.},
  {\sc and} {\sc Spanjaard, O.} 2011.
\newblock Strategy-proof mechanisms for {Facility Location} games with many
  facilities.
\newblock In {\em Proc. of the 2nd International Conference on Algorithmic
  Decision Theory (ADT~'11)}. LNAI Series, vol. 6992. 67--81.

\bibitem[\protect\citeauthoryear{Feldman and Wilf}{Feldman and
  Wilf}{2011}]{FW11}
{\sc Feldman, M.} {\sc and} {\sc Wilf, Y.} 2011.
\newblock Randomized strategyproof mechanisms for {Facility Location} and the
  mini-sum-of-squares objective.
\newblock {\em CoRR abs\/}~{\em 1108.1762}.

\bibitem[\protect\citeauthoryear{Fotakis and Tzamos}{Fotakis and
  Tzamos}{2012}]{FT12}
{\sc Fotakis, D.} {\sc and} {\sc Tzamos, C.} 2012.
\newblock On the power of deterministic mechanisms for facility location games.
\newblock {\em CoRR\/}~{\em abs/1207.0935}.

\bibitem[\protect\citeauthoryear{Fotakis and Tzamos}{Fotakis and
  Tzamos}{2013}]{FT10}
{\sc Fotakis, D.} {\sc and} {\sc Tzamos, C.} 2013.
\newblock Winner-imposing strategyproof mechanisms for multiple {Facility
  Location} games.
\newblock {\em Theoretical Computer Science\/}~{\em 472}, 90--103.

\bibitem[\protect\citeauthoryear{Koutsoupias}{Koutsoupias}{2011}]{Kouts11}
{\sc Koutsoupias, E.} 2011.
\newblock Scheduling without payments.
\newblock In {\em Proc. of the 4th International Symposium on Algorithmic Game
  Theory (SAGT~'11)}. LNCS Series, vol. 6982. 143--153.

\bibitem[\protect\citeauthoryear{Lu, Sun, Wang, and Zhu}{Lu
  et~al\mbox{.}}{2010}]{LSWZ10}
{\sc Lu, P.}, {\sc Sun, X.}, {\sc Wang, Y.}, {\sc and} {\sc Zhu, Z.} 2010.
\newblock {Asymptotically Optimal Strategy-Proof Mechanisms for Two-Facility
  Games}.
\newblock In {\em Proc. of the 11th {ACM} Conference on Electronic Commerce
  (EC~'10)}. 315--324.

\bibitem[\protect\citeauthoryear{Mirchandani and Francis}{Mirchandani and
  Francis}{1990}]{MF90}
{\sc Mirchandani, P.} {\sc and} {\sc Francis, R.} 1990.
\newblock {\em Discrete Location Theory}.
\newblock Wiley.

\bibitem[\protect\citeauthoryear{Miyagawa}{Miyagawa}{2001}]{Miy01}
{\sc Miyagawa, E.} 2001.
\newblock Locating libraries on a street.
\newblock {\em Social Choice and Welfare\/}~{\em 18}, 527--541.

\bibitem[\protect\citeauthoryear{Moulin}{Moulin}{1980}]{Moul80}
{\sc Moulin, H.} 1980.
\newblock On strategy-proofness and single-peakedness.
\newblock {\em Public Choice\/}~{\em 35}, 437--455.

\bibitem[\protect\citeauthoryear{Nisan, Roughgarden, Tardos, and
  Vazirani}{Nisan et~al\mbox{.}}{2007}]{AGT-book}
{\sc Nisan, N.}, {\sc Roughgarden, T.}, {\sc Tardos, E.}, {\sc and} {\sc
  Vazirani, V.} 2007.
\newblock {\em Algorithmic Game Theory}.
\newblock Cambridge University Press.

\bibitem[\protect\citeauthoryear{Nissim, Smorodinsky, and Tennenholtz}{Nissim
  et~al\mbox{.}}{2012}]{NST10}
{\sc Nissim, K.}, {\sc Smorodinsky, R.}, {\sc and} {\sc Tennenholtz, M.} 2012.
\newblock Approximately optimal mechanism design via {Differential Privacy}.
\newblock In {\em Proc. of the 3rd Conference on Innovations in Theoretical
  Computer Science (ITCS~'12)}. 203--213.

\bibitem[\protect\citeauthoryear{Procaccia and Tennenholtz}{Procaccia and
  Tennenholtz}{2009}]{PT09}
{\sc Procaccia, A.} {\sc and} {\sc Tennenholtz, M.} 2009.
\newblock Approximate mechanism design without money.
\newblock In {\em Proc. of the 10th {ACM} Conference on Electronic Commerce
  (EC~'09)}. 177--186.

\bibitem[\protect\citeauthoryear{Schummer and Vohra}{Schummer and
  Vohra}{2002}]{SV02}
{\sc Schummer, J.} {\sc and} {\sc Vohra, R.} 2002.
\newblock Strategyproof location on a network.
\newblock {\em Journal of Economic Theory\/}~{\em 104}, 405--428.

\bibitem[\protect\citeauthoryear{Sprumont}{Sprumont}{1995}]{Spru95}
{\sc Sprumont, Y.} 1995.
\newblock Strategyproof collective choice in economic and political
  environments.
\newblock {\em The Canadian Journal of Economics\/}~{\em 28,\/}~1, 68--108.

\bibitem[\protect\citeauthoryear{Sui, Boutilier, and Sandholm}{Sui
  et~al\mbox{.}}{2012}]{SBS12}
{\sc Sui, X.}, {\sc Boutilier, C.}, {\sc and} {\sc Sandholm, T.} 2012.
\newblock Analysis and optimization of multi-dimensional percentile mechanisms.
\newblock In {\em Proc. of 4th International Workshop on Computational Social
  Choice (COMSOC~'12)}.

\end{thebibliography}

\end{document}